\documentclass[]{article}

\usepackage{booktabs} 
\usepackage{amsmath,amssymb,amsthm}
\usepackage{fullpage,color,xcolor}

\usepackage{bm}
\usepackage{mathrsfs}

\usepackage{amsmath,amsfonts}

\definecolor{DarkBlue}{rgb}{0,0.1,0.55}

\numberwithin{equation}{section}

\newcommand {\hide}[1]{}
\newtheorem{remark}{Remark}

\newcommand {\junk}[1]{}

\newcommand {\R} {\mathbb{R}}

 \newcommand {\N}         {\mathbb{N}}
\newcommand {\Q}         {\mathbb{Q}}

\newcommand {\E} {{\rm Ext}}

\def\addots{\mathinner{\mkern1mu
\raise1pt\vbox{\kern7pt\hbox{.}}
\mkern2mu\raise4pt\hbox{.}\mkern2mu
\raise7pt\hbox{.}\mkern1mu}}

\newcommand{\bmx}{{\bm{x}}}
\newcommand{\bmy}{{\bm{y}}}
\newcommand{\bmz}{{\bm{z}}}
\newcommand{\bmt}{{\bm{t}}}
\newcommand{\bmh}{{H}}
\newcommand{\bfa}{\mathbf{a}}
\newcommand{\jac}{\mathsf{jac}}
\newcommand{\x}{\mathbf{x}}

\newcommand{\mA}{{{\mathbf{A}}}}

\DeclareMathOperator{\mor}{Mor}

\def\Q {\bm{Q}}
\def\R {\bm{R}}
\def\C {\bm{C}}

\def\g {\bm{g}}

\def\mA {\bm{A}}

\def\Sym{{\mathscr{S}}}
\def\VN{{\mathscr{N}}}
\DeclareMathOperator{\Zer}{\ensuremath{V_{\bm{R}}}}

\def\assR{{\mathsf{R}}}
\def\assA{{\mathsf{A}}}

\newenvironment{myitemize}
{ \begin{itemize}
    \setlength{\itemsep}{0pt}
    \setlength{\parskip}{0pt}
    \setlength{\parsep}{0pt}     }
{ \end{itemize}                  }

\newenvironment{myenumerate}
{ \begin{enumerate}
    \setlength{\itemsep}{0pt}
    \setlength{\parskip}{0pt}
    \setlength{\parsep}{0pt}     }
{ \end{enumerate}                  }

\def\added#1{\textcolor{red}{#1}}
\def\added#1{{#1}}

\newtheorem{definition}{Definition}
\newtheorem{theorem}[definition]{Theorem}
\newtheorem{corollary}[definition]{Corollary}
\newtheorem{proposition}[definition]{Proposition}
\newtheorem{lemma}[definition]{Lemma}
\newtheorem{example}[definition]{Example}

\begin{document}
\title
{
Real root finding for equivariant semi-algebraic systems
}

\author{Cordian Riener \\
Department of Mathematics and Statistics\\
UiT The Arctic University of Norway\\
{\tt cordian.riener@uit.no}\\
Mohab Safey El Din\\
Sorbonne Universit\'e, \textsc{CNRS}, \textsc{INRIA},\\
     Laboratoire d'Informatique de Paris~6, \textsc{LIP6},
     \'Equipe \textsc{PolSys}\\
{\tt mohab.safey@lip6.fr}
}




\date{\today}
\maketitle


\begin{abstract}
  Let $\R$ be a real closed field. 
  We consider basic semi-algebraic sets defined by $n$-variate
  equations/inequalities of $s$ symmetric polynomials and an equivariant family
  of polynomials, all of them of degree bounded by $2d < n$. Such a
  semi-algebraic set is invariant by the action of the symmetric group. We show
  that such a set is either empty or it contains a point with at most $2d-1$
  distinct coordinates. Combining this geometric result with efficient
  algorithms for real root finding (based on the critical point method), one can
  decide the emptiness of basic semi-algebraic sets defined by $s$ polynomials
  of degree $d$ in time $(sn)^{O(d)}$. This improves the state-of-the-art which
  is exponential in $n$. When the variables $x_1, \ldots, x_n$ are quantified
  and the coefficients of the input system depend on parameters
  $y_1, \ldots, y_t$, one also demonstrates that the corresponding one-block
  quantifier elimination problem can be solved in time $(sn)^{O(dt)}$.
\end{abstract}

\section{Introduction}\label{sec:intro}

Let $\bm{R}$ be a real closed field.  A \emph{semi-algebraic set} is a
subset of $\R^n$ defined by a boolean formula whose atoms are
polynomial equalities and inequalities with coefficients in $\R$.  In
this article, we consider basic semi-algebraic sets defined as
follows. Given $F = (f_1,\ldots,f_k)$ and $G = (g_1, \ldots, g_s)$ in
$\R[x_1,\ldots,x_n]$, we denote by $S(F, G)\subset \R^n$ the
semi-algebraic set defined by
$f_1= \cdots = f_k =0, g_1 \geq 0, \ldots, g_s\geq 0$.  These sets
arise in many areas of engineering sciences such as computational
geometry, optimization, robotics (see e.g. \cite{Canny, ELLS09,TY87,
  lasserre09}).  Algorithmic problems encompass real root finding,
connectivity queries, or quantifier elimination.

Such problems are intrinsically hard \cite{blum}. In the worst case, solving
quantifier elimination over the reals is doubly exponential in $n$ and
polynomial in the maximum degree of the input polynomials, see \cite{DH88}. This
complexity is achieved by the Cylindrical Algebraic Decomposition algorithm
\cite{BD07}. The idea of reducing real root finding to polynomial optimization
in \cite{Seidenberg} is used in \cite{GV88} to obtain the first algorithm with
singly exponential complexity in $n$. This led to improvements for the
decision problem \cite{Canny1988, HRS93, Ren, BPR98}, quantifier elimination
\cite{HRS90, BPR96, HS12} and connectivity queries \cite{Canny,
  Canny2, HRSRoadmap, GRRoadmap, BaPoRo96, SaSc17}. Later, polar varieties are
introduced in \cite{BGHM97} for the decision problem \cite{BGHM01, BGHP05,
  SaSc03, BGHP10, BGHS14}, for computing roadmaps \cite{SaSc17} or polynomial
optimization \cite{GSZ10, GS14, BGHS14, GGSZ}.  Complexity
bounds are then cubic in some B\'ezout bound as well as
practically efficient algorithms.

To break this curse of dimensionality, one exploits algebraic properties of
systems defining semi-algebraic sets arising in applications. This
has led to improvements, for e.g. the quadratic case \cite{barvinok93, GP05},
the multi-homogeneous case \cite{BFJSV16, HNS16} and the important case of {\em
  symmetric semi-algebraic sets}.

Let $\Sym_n$ denote the group of permutations on a set of cardinality $n$. This
group acts on $\R^n$ by permuting the coordinates. One says that a subset of
$\R^n$ is symmetric when it is closed under this action. 

Let now $f \in \R[x_1, \ldots, x_n]$. One says that $f$ is invariant under the
action of $\Sym_n$ (or in short $\Sym_n$-invariant) when for all
$\sigma \in \Sym_n$, $f(\sigma x) =f$ for $x = (x_1, \ldots, x_n)$. The
following result summarizes the current state-of-the-art on symmetric
semi-algebraic sets. 
\begin{theorem}[\cite{Rie, Rie2, tim}]\label{thm:oldies}
  Let $\{f, \added{f_1,\ldots,f_s}\}\subset\R[x_1,\ldots,x_n]$ be
  $\Sym_n$-invariant polynomials of degree at most $d$.
\begin{myenumerate}
\item\label{thm:oldies:1} The real algebraic set $\Zer(f)$ is not
  empty if and only if it contains a point with at most
  $\lfloor\frac{d}{2}\rfloor$ distinct coordinates.
\item\label{thm:oldies:2} The semi-algebraic set in $S\subset \R^n$
  defined by $\added{f_1\geq 0,\ldots, f_s\geq 0}$ is not empty if and only if
  it contains a point with at most $d$ distinct coordinates.
\end{myenumerate}
\end{theorem}
As a consequence, on input $f$, one can decide the emptiness of
$\Zer(f)$ by partitioning -- up to symmetry -- the set of variables
$x_1, \ldots, x_n$ into $\lfloor\frac{d}{2}\rfloor$ subsets, say
$\chi_1, \ldots, \chi_{\lfloor\frac{d}{2}\rfloor}$ and set $x_i=x_j$
in the input when $x_i$ and $x_j$ lie in the same set
$\chi_\ell$. This way one is led to apply the aforementioned
algorithms for deciding the emptiness of semi-algebraic sets to inputs
involving at most $\lfloor\frac{d}{2}\rfloor$. Since the number of
such partitions lies in $O({n^d})$, one finally obtains algorithms
deciding the emptiness of $\Zer(f)$ (resp. $S$) in time $n^{O(d)}$,
hence polynomial time when $d$ is fixed. The same reasoning holds for
the case \eqref{thm:oldies:2} of Theorem~\ref{thm:oldies}.  

\added{Of course, semi-algebraic sets defined by $\Sym_n$-invariant constraints
  define symmetric semi-algebraic sets but the reciprocal is not true as
  illustrated with the example $x_1\geq 0, x_2\geq 0$.} \added{A family $F$ of
  constraints in $\R[x_1, \ldots, x_n]$ is said to be $\Sym_n$-invariant when
  for all $f\in F$ and $\sigma\in \Sym_n$, $f(\sigma x)\in F$. Note that such
  families of constraints defined symmetric semi-algebraic sets. It is a major
  and longstanding challenge to obtain algorithms that, given a
  $\Sym_n$-invariant family of constraints, takes advantage of the symmetry
  invariance to decide if it is feasible over the reals.}

The goal of this article is to {\em generalize} the results in \cite{Rie, Rie2,
  tim} to the following \added{special} situation. Let $F = (f_1, \ldots, f_k)$
and $G = (g_1, \ldots, g_n)$ be in $\R[x_1, \ldots, x_n]$ and $d$ be the maximum
degree of those polynomials. Assume that for $1\leq i \leq k$, $f_i$ is
$\Sym_n$-invariant and that $G$ is $\Sym_n$-equivariant, i.e., we have
$G(\sigma(x))=(g_{\sigma(i)}(x))_{1\leq i \leq n}$ for all $\sigma \in \Sym_n$.
\added{The topical question we address is the following one.} Can we decide the
emptiness of the semi-algebraic set defined by
$f_1 =\cdots=f_k =0, g_1 \geq 0, \ldots, g_n\geq 0$ in time $n^{O(d)}$, i.e.
polynomial in $n$ and exponential in $d$? \added{More generally, can we take
  advantage of equivariance for e.g. one-block quantifier elimination?}

\added{This latter question is important for a wide range of applications, in
  particular for the analysis of equivariant dynamical systems, which commonly
  appear in biology (see \cite{Stewart2003}). Those systems are of the form
  $\dot{\bmx}=\bm{f}(\bmx, \lambda)$ where $\lambda$ is a set of parameters and
  $\bm{f}$ is an equivariant family of polynomials for the action of the
  symmetric group on the $\bmx$ variables and the state variables $\bmx$ must be
  non-negative (see \cite[Example 2]{Stewart2003}). When analyzing the
  equilibrium points of such systems w.r.t. parameters $\lambda$, we are led to
  solve equivariant semi-algebraic systems.}

\emph{Main results.}  We provide a positive answer to this
question. More precisely, the following holds. \\
{\em (i)} On input $F$ and $G$ as above, deciding the emptiness of
  $S(F, G)\cap \R^n$ can be done in time $n^{O(d)}$.\\
{\em (ii)} On input $F$ and $G$ in $\R[y_1, \ldots, y_t][x_1, \ldots, x_n]$
  such that $F$ and $G$ satisfy the above $\Sym_n$-invariance and
  equivariance assumptions, the quantifier elimination problem
  $\exists x \in \R^n \; F = 0, G >0$ can be solved in time
  $n^{O(dt)}$.

  This result, which generalized Theorem~\ref{thm:oldies}, is of particular
  interest on families of systems where $d$ is fixed and $n$ grows. They are
  obtained by proving that $S(F, G)$ is not empty if and only if it contains a
  point with at most $2d-1$ distinct coordinates. A key ingredient to establish
  such a property is the use of representation theory for equivariant maps and
  basic results from polynomial optimization. Combining such a geometric result
  with efficient algorithms for real root finding or one-block quantifier
  elimination yields the above complexity results. More accurate complexity
  results (with explicit constants in the exponent) are given under some
  assumptions which are proved to be generic.

We also report on practical experiments illustrating that algorithms described
in this paper can tackle semi-algebraic systems which are out of reach of the
current state-of-the-art. \added{}

{\em Related works.} We already mentioned several previous works which led to
Theorem~\ref{thm:oldies}. More generally, the question of using symmetry in the
context of real algebraic geometry is not new. Fundamental work started with
\cite{procesi,PS85} which study the quotient of semi-algebraic sets, which are
invariant under the action of a compact Lie group. In particular,
Positivstellens\"atze for invariant polynomials which are non-negative on
invariant semi-algebraic sets are derived. This line of work is further
generalized by \cite{bro} and was applied for example in \cite{CKS09} to the
context of the moment problem. A different line of work initiated by
\cite{gatermann} consists in exploiting symmetries in the context of sums of
squares relaxations of polynomial optimization. In particular, for
optimization problems which are invariant by the symmetric group, a variety of
strategies are exhibited in \cite{RTJL}. The topology of semi-algebraic sets
defined by symmetric polynomials is also easier to understand: \cite{BS,BR2}
derived efficient algorithms to calculate e.g. their Euler-Poincar\'e
characteristic.

With a more algebraic flavour, computer algebra has been developed to solve
polynomial systems which are invariant under the action of some groups.
Approaches for this longstanding problem focus on the zero-dimensional case and
aim at describing algebraically the solution set. When all equations are
invariant, invariants can be used for this purpose \cite{Colin, Sturmfels08}.
Such an approach is completed by the use of SAGBI Gr\"obner bases techniques
\cite{FR09, Thiery01}. When the system is globally invariant, \cite{FS12, FS13}
provides an efficient dedicated Gr\"obner basis algorithm (see also
\cite{STEIDEL201372,buse16} for further developments).

\emph{Structure of the paper.} Section~\ref{sec:prelim} recalls properties
of symmetric semi-algebraic sets and explains why a direct generalization of
Theorem~\ref{thm:oldies} is hopeless. Section~\ref{sec:geometry} provides a
proof that $S(F, G)$ is not empty iff it contains a point with at most $2d-1$
distinct coordinates. Section~\ref{sec:algo} provides a description of the
algorithms and the analysis of their complexity. Section~\ref{sec:experiments}
reports on practical performances.

\section{Preliminaries}\label{sec:prelim}


A crucial condition in Theorem \ref{thm:oldies} is that all the
polynomials defining the considered semi-algebraic sets are indeed
symmetric. Such an assumption is easily bypassed in the case of
real algebraic sets. 
\begin{corollary}
  Let $f_1,\ldots,f_k\in\R[X_1,\ldots,X_n]$ with $\deg f_i\leq d$ for
  all $i$. Then $\Zer(f_1,\ldots,f_n)$ is not empty if and only if it
  contains a point with at most $d$ distinct coordinates. 
\end{corollary} 
\begin{proof}
  Consider the polynomial
  $g:=\sum_{i=1}^k\sum_{\sigma\in S_n} \sigma(f_i)^2$. Then we have an
  equality of the real varieties $\Zer(g)=\Zer(f_1,\ldots,f_n)$ and
  $g$ is symmetric of degree $2d$ and in this situation statement
  (\ref{thm:oldies:1}) in Theorem \ref{thm:oldies} yields the result.
\end{proof}
However, in many applications the semi-algebraic set is
defined by polynomials that are not themselves symmetric. 

It is feasible to replace the non-invariant inequalities by a set of
new inequalities which describe the same set but are invariant
\cite{bro}: any symmetric semi-algebraic set defined by $s$
inequalities can be defined with $s+1$ inequalities which are
invariant by the action of the symmetric group (see also
\cite{CKS09,hubert1} for a constructive approach).  However, such a
``symmetrization'' process comes at a price: In general it will
increase the degree of the polynomials drastically. We illustrate this
phenomenon in the following easy example.

\begin{example}
  Consider the positive orthant
  $S:=\{\bmx = (\x_1, \ldots, \x_n)\in\R^n\,:\ \x_1\geq 0,\ldots,
  \x_n\geq 0\}$.
  Clearly, $S$ is symmetric. By considering the map which sends the
  coordinates of $\bmx\in\R^n$ to the coefficients of the polynomial
  $h(t):=\prod_{i=1}^n(t-\x_i)$ one can prove that $S$ is equivalently
  defined by $ e_1(x)\geq 0,\ldots, e_n(x)\geq 0$ where $e_i$ denotes
  the $i$-th elementary symmetric polynomial.  One implication is
  immediate: $\bmx\in S$ clearly entails that $e_i(\bmx)\geq 0$ for
  $1\leq i\leq n$.  We prove the other implication by induction on
  $n$. The case $n=1$ is clear.  Now let
  $\bmx=(\x_1, \ldots, \x_n)\in\R^n$ be not in $S$, i.e., let one of its
  coordinates be negative.  Besides, without loss of generality we can
  assume that all $\x_i\neq 0$ (since $S$ has non-empty interior). We
  show further that this implies that there exists $1\leq i \leq n$
  such that $e_i(\bmx) < 0$. Clearly, in the case when  exactly one coordinate of $\bmx$  is
  negative we have $e_n(\bmx)<0$ and hence our
  claim follows.  Therefore, we assume that at least two coordinates are negative.
  We consider the polynomial $h(t)$ as defined above. i.e.,
  $h(t)=t^n+\sum_{i=1}^n(-1)^ie_i(\bmx) t^{n-i}$.  Notice that, by
  construction, all roots of $h$ are real. By Rolle's Theorem, there
  exists a root of its derivative $h':=\frac{\partial h}{\partial t}$
  between every two roots of $h$. Since $h$ has by construction at least two negative roots, $h'$ has a negative root. 
 Consider $(\tilde{\x}_1,\ldots,\tilde{\x}_{n-1})\in\R^{n-1}$ the
  $n-1$ roots of $h'$ (ordered decreasingly).  Since $(\tilde{\x}_1,\ldots,\tilde{\x}_{n-1})$ is not in the $n-1$ dimensional positive orthant we can apply the induction 
  hypothesis to the case $n-1$ to infer that at least for one
  $j\in\{1,\ldots,n-1\}$ we have
  $e_j(\tilde{\x}_1,\ldots,\tilde{\x}_{n-1})<0$. But since $h'$ is the
  derivative of $h$ this clearly implies
  $e_j(\x_1,\ldots,
  \x_n)=\frac{1}{n-j}e_j(\tilde{\x}_1,\ldots,\tilde{\x}_{n-1})<0$.

Of course, the description of $S$ with symmetric
  polynomials is not unique. However, it follows from the equivalence
  shown above that no other description with symmetric polynomials can
  involve only symmetric polynomials of degree smaller than n.

  Indeed, suppose that
  $S:=\{\bmx\in\R^n\,:\ g_1(\bmx)\geq 0,\ldots, g_m(\bmx)\geq 0\}$, where each
  $g_i$ is a symmetric polynomial. It is classically known that each symmetric
  polynomial can be uniquely represented as a polynomial in the elementary
  symmetric polynomials, i.e. for each $i$ we have a polynomial
  $\gamma_i\in\R[e_1,\dots, e_n]$ such that
  $g_i(x)=\gamma_i(e_1(x),\ldots, e_n(x))$. Now suppose that for each $i$ we
  have $\deg g_i < n$. Since the polynomials $\gamma_i$ are unique and
  $\deg e_n= n$, it follows, that for each $i$ we must have
  $\gamma_i(0,0,\ldots,0,t)=0$. Consider the point $\xi:=(0,1,2,\ldots,n-1)$.
  Similarly to the above reasoning, we consider a univariate polynomial
  $h(t):=\prod_{i=1}^n(t-\xi_i)$ (with $\xi_i = i-1$). Note that $e_n(\xi)=0$.
  Since all $n$ roots of $h$ are distinct, $h- \varepsilon$ has also $n$
  distinct real roots, for a small enough positive $\varepsilon$. Let
  $\zeta\in\R^n$ be one of the roots of $h- \varepsilon$. Then, $e_n(\zeta)<0$
  and thus $\zeta\not\in S$. But $e_i(\xi)=e_i(\zeta)$ for all $1\leq i\leq n-1$
  and we deduce that $\gamma_j(\zeta)=\gamma_j(\xi)$ for $1\leq j \leq m$.
  Hence, we get a contradiction with $\zeta \notin S$. Therefore, every
  representation of $S$ in terms of symmetric polynomials must contain at least
  one polynomial of degree $n$, hence making useless Theorem \ref{thm:oldies}
  for algorithmic applications.
\end{example}
Notice that the semi-algebraic set $S$ defined in the example above
clearly contains points for which all coordinates are the same and we
note the following generalization of Theorem \ref{thm:oldies} to basic
convex semi-algebraic sets.
\begin{proposition}\label{prop:convex}
  Let $S\subset\R^n$ be basic convex symmetric semi-algebraic
  set. Then $S$ is not empty if and only if it contains a point for
  which all coordinates are equal.
\end{proposition}
\begin{proof}
  Suppose that $S$ is not empty and let $x\in S$. Since $S$ is
  symmetric, $S$ also contains the orbit
  $\{ \sigma (x)\,:\, \sigma\in \Sym_n\}$ of $x$. Since $S$ is
  convex, it contains the point
  $y:=\frac{1}{n!}\sum_{ \sigma\in \Sym_n} \sigma(x)$ and clearly all
  coordinates of $y$ are equal. 
\end{proof}
Notice that the semi-algebraic set $S$ defined in the example above contains
points for which all coordinates are the same. In view of
Proposition~\ref{prop:convex} it is natural to ask, to which extent it is
possible to derive statements similar to Theorem \ref{thm:oldies} for symmetric
semi-algebraic sets that are defined by polynomials of low degree, which are not
invariant by the action of the symmetric group. The following example shows that
for general semi-algebraic sets, such a generalization is not possible:

\begin{example}
  Let $f:=\sum_{i=1}^{n} (x_i-i)^2$ and its $\Sym_n$ orbit which we denote
  by $\mathcal{F}$. Let $S$ be the semi-algebraic set
  $\{x\in\R^n \,:\, \exists g\in\mathcal{F} \text{ with } g(x)=0\}$. By
  construction, $S$ is a finite set which coincides with the orbit of 
  $\xi=(1,\ldots,n)$. Therefore, all points in $S$ have distinct coordinates,
  but $S$ is described by quadratic polynomials.
\end{example}

\section{Main geometric result}\label{sec:geometry}
One way to generalize Theorem \ref{thm:oldies} to semi-algebraic sets that are
$\Sym_n$-invariant but not described by symmetric polynomials is to rely on
results from the theory of finite reflection groups. A finite group is called a
finite reflection group, if it is generated by orthogonal reflection on a finite
set of hyperplanes. These groups are extensively studied and the particular case
of the symmetric group acting by permuting the coordinates falls into this
framework. We refer the interested reader to \cite{GB96} for more details.

\begin{definition}\label{def:eq}
  Let $\phi: \R^n\rightarrow \R^n$ be a morphism given by
  $\bmx\mapsto (\phi_1(\bmx),\ldots,\phi_n(\bmx))$ and let $G$ be a
  finite reflection group. Then $\phi$ is $G$-equivariant if we have
  $g(\phi) =\phi(g(x))$ for every $g\in G$ . We will write
  $\mor_{G}(\R^n,\R^n)$ for the set of $G$-equivariant morphisms.

  We say that a sequence of polynomials of cardinality $n$ is $G$-equivariant,
  if it defines a $G$-equivariant morphism.
\end{definition}

\begin{example}
  Let $s$ be a bivariate symmetric polynomial and $d\in \N$. The map
  $(\x_1, \x_2, \x_3) \to (\x_1^d+s(\x_2, \x_3), \x_2^d+s(\x_1, \x_3),
  \x_3^d+s(\x_1, \x_2))$
  is equivariant by the action of the symmetric group $\Sym_3$.
\end{example}

Let $\R[x_1, \ldots, x_n]^G$ be the ring of polynomials in
$\R[x_1, \ldots, x_n]$ which are $G$-invariant. There is a natural action of
$\R[x_1, \ldots, x_n]^{G}$ on the set $\mor_{G}(\R^n,\R^n)$ by multiplication:
it clearly preserves the equivariance. In other words, the equivariant morphisms
form a module over $\R[x_1, \ldots, x_n]^{G}$. It follows from the work of
Shchvartsman \cite{Shc} that this module is a free module.

\begin{theorem}[Shchwartsman]\label{thm:sw}
  For any finite reflection group $G$ the set $\mor_{G}(\R^n,\R^n)$ is
  a finite $\R[x_1, \ldots, x_n]^G$-module of rank $n$. Furthermore,
  let $\psi_1,\ldots, \psi_n$ be generators of
  $\R[x_1, \ldots, x_n]^{G}$, then every equivariant morphism can
  uniquely be written as
  $f_i=\sum_{j=1}^n \frac{\partial \psi_j}{\partial x_i} s_j$ where
  $s_j\in\R[x_1, \ldots, x_n]^{G}$.
\end{theorem}


In the sequel we will be interested in basic semi-algebraic sets that
are generated by polynomials which are $\Sym_n$ equivariant in the
sense of Definition \ref{def:eq}. The general machinery developed by
Shchwartsman allows in the case of $\Sym_n$ for the following
corollary, which gives a convenient description of such polynomials.
\begin{corollary}\label{cor:sw}
  Let $\{f_1,\ldots, f_n\}$ be a set of polynomials that define an
  $\Sym_n$ equivariant morphism and let $\deg f_i\leq d$.  Then
  $f_i=\sum_{j=0}^{d} s_j \cdot x_i^{j}$, where
  $s_j\in\R[x_1, \ldots, x_n]^{\Sym_n}$ is symmetric and of
  degree $\leq d-j+1$.
\end{corollary}
\begin{proof}
  It is classically known that every symmetric polynomial can be uniquely
  written in terms of the first $n$ Newton sums $p_i:=\sum_{j=1}^n x_j^i$. Thus,
  we can use these polynomials as generators of $\R[x_1, \ldots, x_n]^{\Sym_n}$
  and apply Theorem \ref{thm:sw}. Since $\mor_{\Sym_n}(\R^n,\R^n)$ is a free
  module and the polynomials $p_1,\ldots, p_n$ are algebraically independent,
  the degree restrictions follow at once, since we cannot have any cancellation
  of degrees in the representation.
\end{proof}

Let us denote by $A_{2d-1}\subset \R^n$ the subset of points with at
most $2d-1$ distinct coordinates.  
\begin{theorem}\label{thm:main1}
  Let $F = (f_1, \ldots, f_k)$ and $G = (g_1,\ldots,g_n)$ be sequences
  of polynomials in $\R[x_1, \ldots, x_n]$. Let $d$ be the maximum of
  $\deg(f_i)$ and $\deg(g_j)$ for $1\leq i \leq k$ and
  $1\leq j \leq n$.

  Assume that for $1\leq i \leq k$, $f_i$ is $\Sym_n$ invariant, that
  $G$ is $\Sym_n$-equivariant and that $\deg (g_j) \geq 2$ for
  $1\leq j \leq n$.

  Then, the basic semi-algebraic set $S(F, G)$ is empty if and only
  if $S(F, G)\cap A_{2d-1}= \emptyset$.
\end{theorem}
Recall that $p_i$ denotes the Newton sum $\sum_{i=1}^n x_i^j$.
For the proof of Theorem \ref{thm:main1}, we study some varieties
defined by the $p_i$'s. Let
$\gamma:=(\gamma_1,\ldots,\gamma_d)\in\R^{d}$ then we denote by
$\VN_\gamma$ the real variety
$$\VN_\gamma:=\{x\in\R^n\,:\, p_1(x)=\gamma_1,\dots,
p_d(x)=\gamma_d\}.$$

These varieties will play a crucial role for the proof of
Theorem~\ref{thm:main1}; the following lemma illustrates the \added{importance
  of these sets}. 
\added{
\begin{lemma}\label{lemma:basic}
  Reusing the notations introduced above, consider a $\Sym_n$-invariant
  polynomial $f$ in $\R[x_1, \ldots, x_n]$ of degree $d$. Then $f$ is constant
  over $\VN_\gamma$. 
\end{lemma}
\begin{proof}
  Since $f$ is $\Sym_n$-invariant, one can write it as the composition
  $q(p_1, \ldots, p_n)$ where $q$ is a polynomial in $\R[u_1, \ldots, u_n]$
  ($u_1, \ldots, u_n$ are new variables) and the $p_i$'s are Newton polynomials
  as above. Since $\deg(f)=d$ and $\deg(p_i)=i$, one also deduces that
  $\deg(q, u_j)=0$ for $d+1\leq j \leq n$. This implies that $q$ lies in
  $\R[u_1, \ldots, u_d]$ and our claim follows immediately from the definition
  of $\VN_\gamma$.
\end{proof}
}
Before going further, we first examine
the possible roots of the polynomials $g_j$ in an $\Sym_n$-equivariant system on
the variety $\VN_\gamma$.
 
 \begin{lemma}\label{lem:nice}
   Let $d\leq n$, $\gamma\in\R^{d}$. Consider $(h_1,\ldots, h_n)$ a
   sequence of polynomials of degree at most $d$ in
   $\R[x_1, \ldots, x_n]$ which are $\Sym_n$- equivariant and
   $\xi=(\xi_1, \ldots, \xi_n) \in \VN_\gamma$.  Then, there exist
   $\{\alpha_1,\ldots,\alpha_{t}\}\in \R^t$ with $t \leq d-1$ such
   that $h_i(\xi)=0$ if and only if
   $\xi_i\in\{\alpha_1,\ldots,\alpha_{t}\}$.
 \end{lemma}
 \begin{proof}
   By Corollary~\ref{cor:sw}, there exist symmetric polynomials $s_i$ of degree
   at most $d$ such that
   $h_i:=\sum_{j=1}^{d} s_j\cdot \frac{\partial p_j}{\partial x_i}$, with
   $\deg(s_j) \leq d$ for all $i\in\{1,\ldots,n\}$. Since for $1\leq j \leq d$,
   $\deg(s_j)\leq d$ and $s_j$ is symmetric, it follows that the value of $s_j$
   at $\xi$ is determined by the value of the first $d$ Newton sums at $\xi$
   (Lemma~\ref{lemma:basic}). Let $\gamma_i = p_i(\xi)$ for $1\leq i \leq d$ and
   $\gamma = (\gamma_1, \ldots, \gamma_d)$ ; besides, observe that, since
   $\xi \in \R^n$, we have $\gamma \in \R^d$. This implies that there exist
   $(b_1, \ldots, b_d)\in \R^d$ such that for all
   $\zeta\in \VN_\gamma\subset \R^n$, $s_1(\zeta)=b_1, \ldots, s_d(\zeta)=b_d$.
   For $1\leq i \leq n$, let us define the univariate polynomial
   $\tilde{h}_i= \sum_{j=1}^{d} b_j x_i^{j-1}$. As a consequence, the equality
   $h_i(\zeta) =\tilde{h}_i(\zeta)$ holds for all $\zeta\in \VN_\gamma$.

   Now, consider the univariate polynomial
   $\delta(U):=\sum_{j=1}^{d} b_j U^{j-1}$ and let
   $\{\alpha_1,\ldots,\alpha_{t}\}$ be its roots in $\R$. Since $\delta$ has
   degree $\leq d-1$, we have $t \leq d-1$. Observe that for every point
   $\xi\in \VN_\gamma\subset \R^n$, $h_i(\xi)=0$ iff $\tilde{h}_i(\xi)=0$ and
   that $\tilde{h}_i(\xi) = \delta(\xi_i)$ where $\xi_i$ is the $i$-th
   coordinate of $\xi$. In other words, $h_i(\xi) = 0$ iff
   $\xi_i \in \{\alpha_1,\ldots,\alpha_{t-1}\}$.
\end{proof}

\begin{proof}[Proof of Theorem \ref{thm:main1}]
  Further $S$ denotes $S(F, G)$. Note that it suffices to show that if
  $S\neq\emptyset$ then there exists a point in $S\cap A_{2d-1}$. So we assume
  that $S\neq\emptyset$ and pick $y\in S$. We set
  $p_1(y)=\gamma_1,\ldots, p_d(y)=\gamma_d$ and we consider the corresponding
  real variety $\VN_\gamma$ as defined above. We now take the intersection
  $S':=S\cap \VN_\gamma$. Notice that $d\geq 2$ (by assumption) and hence
  $\VN_\gamma$ is contained in a sphere. Thus, it follows that $S'$ is closed
  and bounded. Further, we slightly abuse notation by using $p_{d+1}$ to denote
  the map $x \to p_{d+1}(x)$ and its restrictions to subsets of $\R^n$.
  Moreover, since $S'$ is closed and bounded, we deduce that $p_{d+1}(S')$ is
  closed and bounded too (see \cite[Theorem 2.5.8]{BCR}). Hence, we deduce that
  there exists $\xi= (\xi_1, \ldots, \xi_n) \in S'$ with the property that
  $p_{d+1}(\xi)$ is maximal among all points in $S'$. We claim that
  $\xi\in A_{2d-1}$.
  
  Let $\{i_1, \ldots, i_\ell\}$ be the set of indices such that
  $g_i(\xi)=0$ if and only if $i\in \{i_1, \ldots, i_\ell\}$.  By
  Lemma \ref{lem:nice} applied to $G = (g_1, \ldots, g_n)$, we deduce
  that there exists $\alpha = (\alpha_1,\ldots,\alpha_{t})\in \R^{t}$
  with $t \leq d-1$ such that for all $i \in \{i_1, \ldots, i_\ell\}$,
  we have $\xi_i \in \{\alpha_1, \ldots, \alpha_t\}$.  Up to
  re-indexing the variables we can assume that
  $\{i_1, \ldots, i_\ell\} = \{n-\ell+1, \ldots, n\}$. For
  $i\in \{n-\ell+1, \ldots, n\}$, we denote by $\kappa(i)$ the integer
  such that $\xi_i =\alpha_{\kappa(i)}$.  This leads us to consider
  the intersection of $S'$ with the affine linear space $H$ of $\R^n$
  defined by
  $x_{n-\ell+1}-\alpha_{\kappa(n-\ell+1)} = \cdots = x_{n} -
  \alpha_{\kappa(n)} = 0$.
  We denote by $S'_\alpha$ the intersection of $S'$ with the
  aforementioned hyperplanes.

  Recall that $\xi$ lies in $S'_\alpha$ and chosen to maximize
  $p_{d+1}$ on $S'$.  Then, $\xi$ also maximizes the restriction of
  $p_{d+1}$ to $S'_\alpha$. Further, by construction, we have that
  $g_i(\xi)>0$ for all $i\in\{1,\ldots,n-\ell\}$. This shows that
  there exists a ball $B$ centered at $\xi$, of radius small enough
  such that the following holds:\\
  {\em (i)} for $i\in\{1,\ldots,n-\ell\}$, $g_i$ does not vanish in
  $B$; \\
  {\em (ii)} the intersection of $B$ with the real algebraic set
  defined by $ f_1 = \cdots = f_k = g_{n-\ell+1} = \cdots = g_n =0 $
  coincides with $S'_\alpha \cap B$. \\
  Remark now that the real algebraic set defined by
  $f_1 = \cdots = f_k =0$ contains $\VN_\gamma$. Also, applying
  Lemma~\ref{lem:nice} to $G$, one deduces that the real algebraic set
  defined by $g_{n-\ell+1} = \cdots = g_n =0$ coincides with the
  affine linear space $H$. We conclude that $S'_\alpha\cap B$ contains
  $\VN_\gamma\cap H$. Besides, observe that $\xi$ lies in
  $\VN_\gamma\cap H$ and recall again that it maximizes the
  restriction of $p_{d+1}$ to $S'_\alpha$. We deduce that $\xi$
  maximizes the restriction of $p_{d+1}$ to $\VN_\gamma\cap H$.

  Now, two situations may occur. Either, at $\xi$, the truncated
  Jacobian matrix associated to $(p_1,\ldots, p_d)$ obtained by
  considering the partial derivatives w.r.t.
  $(x_1, \ldots, x_{n-\ell})$ is full rank or it is not. In both
  cases, since $\xi = (\xi_1, \ldots, \xi_n)$ maximizes the
  restriction of $p_{d+1}$ to $\VN_\gamma\cap H$, one deduces that
  there exists
  $(\lambda_0, \ldots, \lambda_{d}) \in \R^{d+1}-\{\mathbf{0}\}$ such
  that
  $ 0=\lambda_0\frac{\partial p_{d+1}}{\partial x_j}(\xi)-\sum_{i=1}^d
  \lambda_i \frac{\partial p_i}{\partial x_j}(\xi)$
  for $1\leq j \leq n-\ell$.  This is rewritten as
  $0=(d+1)\lambda_0\xi_j^{d}-\sum_{i=1}^d (i)\lambda_i \xi_j^{i-1}$
  for $1\leq j \leq n-\ell$.  The above algebraic relation entails
  that for $1\leq j \leq n-\ell$, $\xi_j$ is a root of the non-zero
  univariate polynomial $\eta (U):=\sum_{i=0}^{d}\lambda_i U^{i}$ of
  degree at most $d$ (recall that
  $(\lambda_0, \ldots, \lambda_d)\neq (0, \ldots, 0)$). Therefore, at
  most $d$ of the first $n-\ell$ coordinates of $\xi$ can be
  distinct. Further, by construction, we have that there are at most
  $d-1$ possibilities for the last $\ell$ coordinates of
  $\xi$. Therefore, $\xi\in A_{2d-1}$ as claimed.
\end{proof}
\added{
  \begin{remark}
    Observe that when $F$ is $\Sym_n$-equivariant (instead of having all of its
    entries $\Sym_n$-invariant), the conclusions of Theorem~\ref{thm:main1}
    still hold. To see that it suffices to replace $F$ be the sum of the squares
    of its entries. Also when the entries of $F$ are subject to inequality
    constraints (instead of equality constraints), the conclusions of
    Theorem~\ref{thm:main1} still hold as one can replace inequalities by
    equations as in \cite[Chap. 13]{BaPoRo96}.
  \end{remark}
}

\section{Algorithms and complexity}\label{sec:algo}

\subsection{Deciding emptiness}

Further, we let $\Q$ be a real field, $\R$ be a real closed field
containing $\Q$ and $\C$ be an algebraic closure of $\R$. We consider
$F = (f_1, \ldots, f_k)$ and $G = (g_1, \ldots, g_n)$ be polynomial
sequences in $\Q[x_1, \ldots, x_n]$. As above, the semi-algebraic set
of $\R^n$ defined by
$$ f_1 = \cdots = f_k =0, \qquad g_1 \geq 0, \ldots, g_n \geq 0 $$ is
denoted by $S(F, G)$.  We start with a first complexity statement.

\begin{theorem}\label{thm:complexity1}
  Let $F$ and $G$ be as above and $d$ be an integer bounding the degrees of the
  polynomials in $F$ and $G$. Assume that the polynomials in $F$ are
  $\Sym_n$-invariant and that the map $\x \mapsto (g_1(\x), \ldots, g_n(\x))$ is
  $\Sym_n$-equivariant and that $d \leq n/2$.
  
  There exists an algorithm which, on input $(F, G)$ decides whether
  $S(F, G)$ is empty using at most $n^{O(d)}$ arithmetic operations in
  $\Q$.
\end{theorem}

\begin{proof}
  By Theorem~\ref{thm:main1}, $S(F, G)$ is non-empty if and only if
  there exists $\bmx\in S(F, G)$ with at most $2d-1$ distinct
  coordinates.

  For $\bmx = (\x_1, \ldots, \x_n)\in \R^n$, we denote by
  $v(\bmx)=\{v_1, \ldots, v_p\}$ (with $p \leq n$ depending on $\bmx$)
  the set of values taken by the coordinates of $\bmx$ and by
  $\mathscr{P}(\bmx) = (\mathscr{P}_1, \ldots, \mathscr{P}_p)$ the
  partition given by the sets
  $\mathscr{P}_j = \{x_i \mid \x_i = v_j\}$.  Up to renumbering, one
  assumes that the $\mathscr{P}_i$'s are given by ascending
  cardinality.

  Hence, set $r = 2d-1$ and consider a partition
  $\gamma = [\gamma_1, \ldots, \gamma_r]$ of $n$ of size $r$,
  i.e. $\gamma_1+\cdots+\gamma_r=n$ with
  $\gamma_i\in \mathbb{N}-\{0\}$ for $1\leq i \leq r$ and
  $\gamma_{i-1}\leq \gamma_i$ (by convention, $\gamma_0=0$). We say
  that a partition $\mathscr{P}_1, \ldots, \mathscr{P}_p$ of
  $(x_1, \ldots, x_n)$ is compatible with $\gamma$ if $p \leq r$ and
  there exists an increasing sequence of integers $s_i$ such that
  $|\mathscr{P}_i|=\gamma_{s_{i-1}}+\cdots+\gamma_{s_{i}}$.
  
  We prove below that, given $\gamma$, one can decide in time
  $n^{O(r)}$ the existence of a real point $\bmx$ in $S(F, G)$ such
  that $\mathscr{P}(\bmx)$ is compatible with $\gamma$. Bounding
  further the number of partitions of size $r$ by $n^{r}$ will
  establish the announced result (recall that $r = 2d-1$).

  Assume that such a point $\bmx=(\x_1, \ldots, \x_n)$ exists and
  consider
  $\mathscr{P}(\bmx) = (\mathscr{P}_1, \ldots, \mathscr{P}_p)$. For
  $\sigma\in \Sym_n$ we denote by $\sigma(\mathscr{P}_i)$ the set
  $\{x_{\sigma(j)} \mid x_j \in \mathscr{P}_i\}$.  Let $\sigma$ be a
  permutation of $\Sym_n$ such that
  $\sigma(\mathscr{P}_i) = \{x_{s_{i-1}}, \ldots, x_{s_{i}}\}$ and
  consider $\sigma(\bmx) = (\x_{\sigma(1)}, \ldots,
  \x_{\sigma(n)})$.
  Now, remark that since all entries of $F$ are invariant by the
  action of $\Sym_n$ and that $G$ is $\Sym_n$-equivariant,
  $\sigma(\bmx) \in S(F, G)$.  This leads us to associate to $\gamma$
  the partition $\Gamma = (\Gamma_1, \ldots, \Gamma_r)$ of
  $(x_1, \ldots, x_n)$ defined by
  $\Gamma_i=\{x_{\gamma_{i-1}+1}, \ldots, x_{\gamma_i}\}$ with
  $\gamma_0=0$ by convention.

  Now, let $a_1, \ldots, a_r$ be new indeterminates. Next, we perform the
  substitution $ x_{\gamma_{i-1}+1} = \cdots = x_{\gamma_{i}} = a_{i}$ for
  $1\leq i \leq r$ in $F$ and $G$. In the end, one obtains polynomial families
  in $\R[a_1,\ldots, a_r]$. The above discussion shows that we only need to
  decide the existence of real points to the new system one obtains this way.
  Using \cite[]{BaPoRo06}, this is done in time $n^{O(r)}$.

  To finish the proof, it remains to count the number of partitions
  $\gamma = (\gamma_1, \ldots, \gamma_r)$ of $n$. The number $p(n, r)$ of
  partitions of $n$ of size $r$ satisfies the recurrence relation
  $p(n, r) = p(n-1, r)+p(n-r, r)$ with $p(n,r)=0$ if $n < r$ and
  $p(n,n)=p(n,1)=1$. A simple induction establishes the inequality
  $p(n, r)\leq n^r$ which finishes the proof.
\end{proof}




  

The next result establishes a more precise complexity statement: we will
actually identify the constant which is in the big-Oh exponent, when
the input system satisfies some properties that we will prove to be
generic. Further, given a polynomial family $H$ in
$\R[x_1, \ldots, x_n]$, $V(H)\subset \C^n$ denotes the set of common
solutions to $H$ in $\C^n$.

Hence, let as above $F=(f_1, \ldots, f_k)$ and $G=(g_1, \ldots, g_n)$ in
$\R[x_1, \ldots, x_n]$. \added{Further, for
  $\mathcal{I}=\{i_1, \ldots, i_\ell\}\subset \{1, \ldots, n\}$, we denote by
  $H_{\mathcal{I}}$ the set $ F\cup \{g_{i_1}, \ldots, g_{i_\ell}\}$.} We say
that $(F, G)$ satisfies the assumption $\assR$ when
\begin{myitemize}
\item the jacobian matrix of $F$ has maximal rank at all points in $V(F)$;
\item \added{for all $\mathcal{I}\subset \{1, \ldots, n\}$, the jacobian matrix of
  $H_{\mathcal{I}}$ has maximal rank at any point of $V(H_{\mathcal{I}})$}.
\end{myitemize}
Now, let $r$ in $\{1, \ldots, n\}$; we say that $(F, G)$ satisfies
assumption $\assA_{r}$ if for any partition
$\gamma = (\gamma_1, \ldots, \gamma_r)$ of $n$, when performing the
substitution $x_{\gamma_{i-1}+1} = \cdots = x_{\gamma_i} = a_i$
($1\leq i \leq r$) where $a_1, \ldots, a_r$ are new variables in
$(F, G)$, the obtained couple of polynomial sequences
$(F_\gamma, G_\gamma)$ satisfies $\assR$.

Further, the entries of $F_\gamma$ (resp. $G_\gamma$) are denoted by
$f_{1, \gamma}, \ldots, f_{k, \gamma}$ (resp.
$g_{1, \gamma}, \ldots, g_{n, \gamma}$).  We can now state our
complexity result.

\begin{theorem}\label{thm:algo}
  Let $F $ and $G$ be as above in $\R[x_1, \ldots, x_n]$, $d$ be the
  maximum of the polynomials in $F$ and $G$ and $E$ be the complexity
  of evaluating $(F, G)$.  Assume that for $1\leq i \leq n$, $f_i$ is
  $\Sym_n$ invariant, that $\g$ is $\Sym_n$-equivariant and that
  $\deg (g_j) \geq 2$ for $1\leq j \leq n$ and that $(F, G)$ satisfies
  assumption $\assA_r$.

  There exists an algorithm which on input $(F, G)$ satisfying
  $\assA$, decides whether $S(F, G)$ is empty using 
  $O{}^{\tilde{~}}\left
    (n^{2d}(2d)^{4d+1}(pE+d^2)\right
  )$
  arithmetic operations in $\Q$.
\end{theorem}

\begin{proof}
  Further we set $r = 2d-1$.  By Theorem~\ref{thm:main1}, $S(F, G)$ is
  not empty if and only if $S(F, G)$ contains a point of $\R^n$ with
  at most $r$ distinct coordinates. Besides, using the invariance of
  $(F, G)$ under the action of $\Sym_n$ as in the proof of
  Theorem~\ref{thm:complexity1}, deciding if $S(F, G)$ contains a real
  point with at most $r$ distinct coordinates can be done by deciding
  if at least one of the semi-algebraic sets
  $S_\gamma = S(F_\gamma, G_\gamma)$ is non-empty when $\gamma$ ranges
  of the set of partitions of $n$ of length $r$. We already
  established that the number of such partitions is upper bounded by
  $n^r$ at the end of the proof of Theorem~\ref{thm:complexity1}.

  Hence, let us focus on the complexity of deciding if $S_\gamma$ is
  empty. We need to introduce some notation. For
  $\mathcal{I} = \{i_1, \ldots, i_\ell\}\subset \{1, \ldots, n\}$, we
  denote by $V_{\gamma, \mathcal{I}}\subset \C^n$ the algebraic set
  defined by
  $f_{1, \gamma} = \cdots = f_{k,\gamma} = g_{i_1, \gamma} = \cdots =
  g_{i_\ell, \gamma} = 0$.
  Further, we denote by
  $\bmh = (h_1, \ldots, h_m)\subset \R[a_1, \ldots, a_r]$ these
  polynomials defining $V_{\gamma, \mathcal{I}}$.  Further, we use
  linear changes of variables. Hence for $\mA\in \mathrm{GL}_n(\R)$,
  we denote by $h_i^{\mA}$ the polynomial obtained by performing the
  change of variables $\bfa \mapsto \mA^{-1} \bfa$ in $h_i$ and by
  $\bmh^{\mA}$ the sequence $(h_1^{\mA}, \ldots, h_m^{\mA})$.

  Using \cite{BaPoRo06}, we deduce that, to decide the emptiness of $S_\gamma$,
  it suffices to compute sample points in each connected component of the real
  algebraic set $V_{\gamma, \mathcal{I}}\cap \R^n$ for all
  $\{i_1, \ldots, i_\ell\}\subset \{1, \ldots, n\}$ and filter out those points
  which lie in $S_\gamma$. Since $(F, G)$ satisfies $\assA$,
  $V_{\gamma, \mathcal{I}}$ is either empty or smooth and equidimensional of
  co-dimension $k+\ell$ and the above polynomial system generates a radical
  ideal (by the Jacobian criterion \cite[Theorem 16.19]{Eisenbud}). We conclude
  that we only need to consider subsets of cardinality $\ell \leq r-k$. We are
  in position to apply the results in \cite{SaSc03}.

  Actually, we use a variant of the algorithm in \cite{SaSc03}, combining the
  geometric approach described therein with \cite{SaSc16}. Let $\pi_i$ be the
  canonical projection
  $(\bfa_1, \ldots, \bfa_r)\mapsto (\bfa_1, \ldots, \bfa_i)$ and, given an
  equidimensional and smooth algebraic set $V$, let $W(\pi_i, V)$ be the
  critical locus of the restriction of $\pi_i$ to $V$. We will also consider the
  projections $\varphi_i: (\bfa_1, \ldots, \bfa_r)\mapsto \bfa_i$.

  By \cite[Theorem 2]{SaSc03}, in order to decide the emptiness of
  $V^{\mA}_{\gamma, \mathcal{I}}\cap \R^n$, it suffices to perform a generic
  linear change of variables $\mA\in \mathrm{GL}_r(\R)$ and next compute
  rational parametrizations of all sets
  $\pi_{i-1}^{-1}(0)\cap W(\pi_{i}, V^{\mA}_{\gamma, \mathcal{I}})$ for
  $1\leq i \leq \dim(V^{\mA}_{\gamma, \mathcal{I}})+1$. Technical but immediate
  computations show that
  $\pi_{i-1}^{-1}(0)\cap W(\pi_{i}, V^{\mA}_{\gamma, \mathcal{I}}) =
  W(\varphi_i, Z_i)$
  where $Z_i = \pi_{i-1}^{-1}(0)\cap V^{\mA}_{\gamma, \mathcal{I}}$. Observe
  that in order to compute
  $Z_i = \pi_{i-1}^{-1}(0)\cap V^{\mA}_{\gamma, \mathcal{I}}$ it suffices to
  solve the so-called Lagrange system
  $h^{\mA}_{1,i-1}=\cdots=h^{\mA}_{m, i-1}=0, [\ell_1, \ldots, \ell_m]
  \jac(\bmh_{i-1}^{\mA}, i) = \mathbf{0}$
  where $h^{\mA}_{j,i-1}$ (resp. $\bmh^{\mA}_{i-1}$) is the polynomial obtained
  by setting $a_1=\cdots=a_{i-1}=0$ in $h^{\mA}_j$ (resp. $\bmh^{\mA}$), and
  $\jac(\bmh^{\mA}_{i-1},1)$ is the submatrix obtained by removing the first
  column of the Jacobian matrix associated with $\bmh^{\mA}_{i-1}$. By
  \cite[Proposition B.1]{SaSc17}, after performing a generic linear change of
  variables, assumptions needed to apply \cite[Theorem 16]{SaSc16}. This latter
  result shows that, letting $E_\bmh$ be the complexity of evaluating $\bmh$ and
  $r_i=r-(i-1)$, one can solve the above Lagrange system using
  $O{}^{\tilde{~}}\left (r_i^3\binom{r_i}{m}d^{2r_i+1}(pE_\bmh+r_id+r_i^2)\right
  )$
  arithmetic operations in $\Q$. Hence, since
  $\dim(V_{\gamma, \mathcal{I}})=r-m$ the total cost of computing sample points
  in each connected component of $V_{\gamma, \mathcal{I}}\cap \R^n$ uses
  $O{}^{\tilde{~}}\left (r^42^{2r}d^{2r+1}(pE_\bmh+rd+r^2)\right )$ arithmetic
  operations in $\Q$. Finally, observe that $E_\bmh$ is bounded by the
  complexity of evaluating the input $(F, G)$. Also, summing up this cost to
  take into account all possible subsets $\mathcal{I}$ (bounded by $2^r$) and
  the number of partitions $\gamma$ (bounded by $n^r$) ends the proof.
\end{proof}

It remains to establish the genericity of assumption $\assA$. To do that, we
need to define the parameters'space in which the genericity statement will hold,
i.e. the space of the coefficients of $(F, G)$ where all entries of $F$ are
$\Sym_n$-invariant and $G$ is $\Sym_n$-equivariant (recall also that all entries
of $(F, G)$ have degree bounded by $d$). Let $\mathfrak{R}$ be the Reynolds
operator which sends $f \in \C[x_1, \ldots, x_n]$ to
$\mathfrak{R}(f) = \frac{1}{n!}\sum_{\sigma\in \Sym_n} \sigma(f)$. Let now
$\mathcal{M}$ be the set of all monomials of degree $\leq d$ in
$\C[x_1, \ldots, x_n]$ and
$\mathcal{M}_{\mathfrak{R}} = \mathfrak{R}(\mathcal{M})$,
$c = |\mathcal{M}_{\mathfrak{R}}|$ and $T_i = (t_{i,1}, \ldots, t_{i,c})$ be new
indeterminates for $1\leq i \leq k$. we define now
$\mathfrak{f}_i = \sum_{m_j \in \mathcal{M}_{\mathfrak{R}}} t_{i,j} m_j$ and
$\mathfrak{f} = (\mathfrak{f}_1, \ldots, \mathfrak{f}_k)$ in
$\C(T_1, \ldots, T_k)[x_1, \ldots, x_n]$. Observe that any sequence $F$ such
that all entries have degree bounded by $d$ and are $\Sym_n$-invariant are
obtained by specializing the indeterminates $(T_1, \ldots, T_k)$. Finally, we
consider an additional sequence of indeterminates
$T_{k+1} = (t_{k+1,1}, \ldots, t_{k+1,c})$ a polynomial
$\mathfrak{g} = \sum_{m_j \in \mathcal{M}_{\mathfrak{R}}} t_{k+1, j} m_j$.
Again, any sequence $G$ which is $\Sym_n$-equivariant is obtained as the
gradient vector of a polynomial obtained by instantiating $T_{k+1}$ in
$\mathfrak{g}$. Then, we set $N = c(k+1)$ and the parameters'space we consider
is $\C^{N}$, i.e. the one endowed by the indeterminates
$(T_1, \ldots, T_{k+1})$.

\begin{theorem}
  There exists a non-empty Zariski open set
  $\mathscr{O}\subset \C^{N}$ such that for $(F, G)$ in $\mathscr{O}$,
  $(F, G)$ satisfies assumption $\assA$.
\end{theorem}

\begin{proof}
  Let $\mathcal{I} = \{i_1, \ldots, i_\ell\}\subset \{1, \ldots, n\}$,
  $r = 2d-1$, $\gamma = (\gamma_1, \ldots, \gamma_r)$ and
  $E_\gamma\subset \C^n$ be the linear subspace defined by
  $x_{\gamma_{i-1}+1} = \cdots = x_{\gamma_i}$ (for $1\leq i \leq
  r$). We denote by $\Gamma$ this set of linear equations which define
  $\E_\gamma$.  We consider the map
  $\Phi_{\mathcal{I}, \gamma}: \bmz = (\bmx, \bmt)\in E_\gamma \times
  \C^{N} \mapsto \left (\mathfrak{f}_1(\bmz), \ldots,
    \mathfrak{f}_k(\bmz), \frac{\partial \mathfrak{g}}{\partial
      x_{i_1}}(\bmz), \ldots, \frac{\partial \mathfrak{g}}{\partial
      x_{i_\ell}}(\bmz)\right )$. Assume for the moment that
  $\mathbf{0}$ is a regular value of $\Phi_{\mathcal{I}}$. Then, the
  algebraic version of Thom's weak transversality theorem (see
  e.g. \cite[Proposition B.3]{SaSc17}) states that there exists a
  non-empty Zariski open set $\mathscr{O}_{\mathcal{I}}$ such that for
  any $\bmt\in \mathscr{O}_{\mathcal{I}}$, $\mathbf{0}$ is a regular
  value for the specialized map
  $\bmx \mapsto \Phi_{\mathcal{I}}(\bmx, \bmt)$. In other words, at
  any $\bmx\in \C^N$ in the zero-set of the union of $\Gamma$ with
  $\mathfrak{f}_1(., \bmt), \ldots, \mathfrak{f}_k(., \bmt),
  \frac{\partial \mathfrak{g}}{\partial x_{i_1}}(., \bmt), \ldots,
  \frac{\partial \mathfrak{g}}{\partial x_{i_\ell}}(.,\bmt)$, the
  Jacobian matrix of that polynomial family is full rank. By the
  Jacobian criterion \cite[Theorem 16.19]{Eisenbud}, we deduce that
  this polynomial family satisfies $\assR$. Finally, we define
  $\mathscr{O}$ as the intersection of the finitely many non-empty
  Zariski open subsets $\mathscr{O}_{\mathcal{I}}\subset \C^N$. Hence,
  $\mathscr{O}$ is a non-empty Zariski open set of $\C^N$ and for any
  $(F, G)\in \mathscr{O}$, $(F, G)$ satisfies $\assA$.

  It remains to prove that for
  $\mathcal{I}=\{i_1, \ldots, i_\ell\}\subset\{1, \ldots, n\}$,
  $\mathbf{0}$ is a regular value of the map
  $\Phi_{\mathcal{I}, \gamma}$, i.e. the Jacobian matrix associated to
  $\Phi_{\mathcal{I}, \gamma}$ is invertible at any point of
  $\Phi_{\mathcal{I}, \gamma}^{-1}(\mathbf{0})$. To do that, we prove
  that the Jacobian matrix associated to $\Gamma$ and
  $\left (\mathfrak{f}_1, \ldots, \mathfrak{f}_k, \frac{\partial
      \mathfrak{g}}{\partial x_{i_1}}, \ldots, \frac{\partial
      \mathfrak{g}}{\partial x_{i_\ell}}\right )$ is full rank at any
  point of $\Phi_{\mathcal{I}, \gamma}^{-1}(\mathbf{0})$. We extract a
  full rank submatrix of that Jacobian matrix as follows:
  \begin{myitemize}
  \item since $\Gamma$ is a set of independent linear equations, one extracts a
    full rank square submatrix $\mathbf{J}$ with entries in $\C$ 
    whose columns correspond to partial derivatives w.r.t. variables in
    $x_1, \ldots, x_n$;
  \item we select the columns corresponding to the partial derivatives
    w.r.t. indeterminates encoding the constant terms in
    $\mathfrak{f}_i$; this yields a diagonal submatrix $\Delta$ with
    $1$'s on the diagonal;
  \item we select the columns corresponding to the partial derivatives
    w.r.t. the indeterminate multiplying $(x_1+\cdots+x_n)$ in
    $\mathfrak{g}$; this yields a diagonal submatrix $\Delta'$, with
    $1$'s on the diagonal.
  \end{myitemize}
  In the end, the submatrix we have extracted is block-diagonal and
  these blocks on the diagonal are $\mathbf{J}$, $\Delta$ and
  $\Delta'$. This ends the proof.
\end{proof}

\subsection{One-block quantifier elimination}

We now study the situation where $F= (f_1, \ldots, f_k)$ and
$G = (g_1, \ldots, g_n)$ are polynomials in
$\Q[x_1, \ldots, x_n, y_1, \ldots, y_t]$ such that \\
{\em (i)} the action of $\Sym_n$ on $(x_1, \ldots, x_n)$ leaves
invariant $f_i$ ; {\em (ii)} the map
$\bmx=(\x_1, \ldots, \x_n)\mapsto (g_1(\bmx, .), \ldots, g_n(\bmx,
.))$ is $\Sym_n$ equivariant.

We consider the problem of computing a semi-algebraic description of the
projection on the $(y_1, \ldots, y_t)$-space of the set
$S(F, G)\subset \R^n\times \R^t$ defined by
$f_1=\cdots=f_k=0, g_1\geq 0, \ldots, g_n\geq 0$. This is equivalent to solve
the one-block quantifier elimination problem:
\begin{equation}
  \label{eq:formulaqe}
\Phi: \quad  \exists \bmx \in \R^n \quad f_1=\cdots=f_k=0,\quad g_1\geq 0, \ldots, g_n\geq 0, 
\end{equation}
hence computing a quantifier-free formula which is equivalent to the
quantified formula $\Phi$. A geometric interpretation is that one aims
at computing a semi-algebraic description of $\Pi(S(F, G))$ where
$\Pi$ is the projection $(\bmx, \bmy)\in \R^n\times \R^t\mapsto \bmy$.

\begin{theorem}\label{thm:qecomplexity1}
  Let $F$, $G$ and $\Phi$ be as above, $d$ be the maximum degree in
  the variables in $(x_1, \ldots, x_n)$ of the entries of $F$ and $G$
  and assume that $d \leq \frac{n}{2}$.  Then, there exists a
  quantifier-free formula $\Psi(Y) = \cup_{k=1}^K\Psi_{k}(Y)$ which is
  equivalent to $\Phi$, and such that $K\leq n^{O(d)}$ and:
  \[
    \Psi_k(Y) = \vee_{i=1}^{\ell_k} \wedge_{j=1}^{\ell_{i,k}}
    (\vee_{u=1}^{\ell_{i,j,k}} {\rm sign}(\varphi_{i,j,u,k}) =
    \sigma_{i,j,h,k})
  \]
  \begin{align*}
\text{with }\quad    \sigma_{i,j,h,k}\in \{ 0, 1, -1\}, & \quad   \ell_k \leq  (n+k)^{d+1}n^{O(dt)}\\
    \ell_{i,k} \leq  (n+k)^{d+1}n^{O(d)}, & \quad \ell_{i,j,k} \leq  n^{O(d)}
  \end{align*}
  and the degrees of the polynomials $\varphi_{i,j,u}$ are bounded by
  $n^d$. Moreover, there exists an algorithm which computes $\Psi$
  using at most $(k+n)^{dt}n^{O(dt)}$ arithmetic operations in $\Q$.
\end{theorem}

\begin{proof}
  The proof is similar to the one of Theorem~\ref{thm:complexity1}. We reduce
  the considered one-block quantifier elimination problem to solving finitely
  many one-block quantifier elimination problems.

  Set $r = 2d-1$ and let $\Gamma(n, r)$ be the set of partitions
  $\gamma = (\gamma_1, \ldots, \gamma_r)$ of $n$ of size $r$ ($\gamma_0 = 0$ by
  convention). As in the proof of Theorem~\ref{thm:complexity1}, we associate to
  $\gamma$ the substitution $ x_{\gamma_{i-1}+1} = \cdots = x_{\gamma_i} = a_i $
  for $1\leq i \leq r$, where $a_1, \ldots, a_r$ are new variables. We denote by
  $\Phi_\gamma$ the formula obtained after performing this substitution in
  $\Phi$ and by $S_\gamma(F, G)$ the semi-algebraic set in $\R^{r}\times \R^t$
  defined by the system obtained after applying the same substitution. Assume,
  for the moment, the following equality:
  \begin{equation}
    \label{eq:3}
    \Pi(S(F, G)) = \cup_{\gamma \in \Gamma(n, r)} \Pi(S_\gamma(F, G)). 
  \end{equation}
  Then, performing quantifier elimination on formula $\Phi$ is
  equivalent to performing quantifier elimination on each formula
  $\Phi_\gamma$ -- which yields a quantifier-free formula
  $\Psi_\gamma$ defining $\Pi(S_\gamma(F, G))$ -- and returning
  $\bigvee_{\gamma \subset \Gamma(n, r)} \Psi_\gamma$.  Using
  \cite[Theorem 14.16]{BaPoRo06}, one deduces that performing
  quantifier elimination on $\Phi_\gamma$ is done using
  $(k+n)^{dt}n^{O(dt)}$ arithmetic operations in $\Q$ and it yields a
  formula
  $ \Psi_\gamma(Y) = \vee_{i=1}^{\ell} \wedge_{j=1}^{\ell_{i}}
  (\vee_{u=1}^{\ell_{i,j}} {\rm sign}(\varphi_{i,j,u}) =
  \sigma_{i,j,h}) $
  such that $\ell \leq (n+k)^{d+1}n^{O(dt)}$,
  $\ell_i \leq (n+k)^{d+1}n^{O(d)}$ and $\ell_{i,j} \leq n^{O(d)}$.
  Now, recall that $\Gamma(n, r)$ has cardinality bounded by
  $n^{O(d)}$ (this bounds the integer $K$ in the statement of the
  Theorem).  Hence, runtime and degree bounds on the output formula
  are established. 
  
  Now, we prove that \eqref{eq:3} holds which will end the proof.  Let
  $\bmy\in \Pi(S(F, G))$ and $S_\bmy\subset \R^n$ be the projection of
  $S(F, G)\cap \Pi^{-1}(\bmy)$ on the $(x_1, \ldots, x_n)$-space.
  Observe that $S_\bmy$ is defined by the polynomial system obtained
  by instantiating variables $(y_1, \ldots, y_t)$ to the coordinates
  of $\bmy$ in $F$ and $G$ ; we denote the obtained polynomial
  sequences by $F_\bmy$ and $G_\bmy$.

  Observe that all entries of $F_\bmy$ are $\Sym_n$-invariant, the sequence
  $G_\bmy$ defines an equivariant map and all entries of $F_\bmy$ and $G_\bmy$
  have degree $\leq \frac{n}{2}$ by assumption. Hence, we can apply
  Theorem~\ref{thm:main1}. It establishes that the semi-algebraic set $S_\bmy$
  is empty if and only if it contains a point $\bmx$ with at most $2d-1$
  coordinates. Now, observe, as in the proof of Theorem~\ref{thm:complexity1},
  thanks to the invariance of $(F_\gamma, G_\gamma)$ that since under the action
  of $\Sym_n$, there exists a partition $\gamma=(\gamma_1, \ldots, \gamma_r)$ in
  $\Gamma(n,r)$ such that $S_\bmy$ has a non-empty intersection with the
  hyperplanes defined by $x_{\gamma_{i-1}+1} = \cdots = x_{\gamma_i}$ for
  $1\leq i \leq r$. In other words, there exists $\gamma \in \Gamma(n, r)$ such
  that $\bmy\in \Pi(S_\gamma(F, G))$. We deduce that
  $\Pi(S(F, G)) \subset \cup_{\gamma \in \Gamma(n, r)}\Pi(S_\gamma(F, G))$. The
  reverse inclusion is immediate once we observe that
  $ \cup_{\gamma \in \Gamma(n, r)}S_\gamma(F, G)\subset S(F, G)$.
\end{proof}

\section{Experimental results}\label{sec:experiments}

Our experiments make use of the following software
\begin{myitemize}
\item {\sc RAGlib.} \cite{raglib}. This is a Maple library, based on the {\sc
    FGb} library by J.-C. Faug\`ere
  . It implements algorithms based on
  the critical point method running in time singly exponential in $n$.
\item {\sc Mathematica-CAD} \cite{Str06}, {\sc RealTriangularize}
  \cite{CDLMXX12} which are packages computing Cylindrical Algebraic
  Decompositions (CAD) adapted to polynomial sequences.
\end{myitemize}
We have considered the following test-suites:
\begin{myitemize}
\item {\bf S1}. We take the gradient of randomly chosen dense
  symmetric polynomials for $G$, letting $F$ be the empty sequence;
  the coefficients of these polynomials are chosen between $-2^{16}$
  and $2^{16}$ using the random tool generator of {\sc Maple}.
\item {\bf S2}. We take random dense systems of symmetric polynomials and
  equivariant families in $\mathbb{Q}[x_1, \ldots, x_n]$.
\end{myitemize}

To solve these polynomial systems, we will use the following
implementations of critical point method-based algorithms and
CAD:
\begin{myitemize}
\item {\bf RAG} refers to the Real Algebraic Geometry library {\sc
    RAGlib}; 
\item {\bf M} refers to the CAD  in {\sc Mathematica}; 
\item {\bf T} refers to the CAD package in {\sc Maple}.
\end{myitemize}
The direct use of these polynomials will be compared with algorithms on which
Theorems~\ref{thm:complexity1} and~\ref{thm:algo}~rely. These consist
in using critical point based algorithms to decide if the input system has a
real solution with at most $2d-1$ distinct coordinates (where $d$ bounds the
degree of the inputs). Also we can substitute the use of those algorithms
by ones that are based on CAD. This leads us to consider in our comparisons the
following:
\begin{myitemize}
\item {\bf RS}: consists in using {\sc RAGlib}; this is an
  implementation of the algorithm on which
  Theorems~\ref{thm:complexity1} and~\ref{thm:algo} rely.
\item {\bf RS-T}: consists in using the CAD Maple package to look at
  solutions with at most $2d-1$ distinct coordinates.
\end{myitemize}
The computations are performed on an Intel(R) Xeon(R) CPU E3-1505M v6
@ 3.00GHz with $32$ Gb of RAM. Timings are given in seconds. The
symbol '-' means that no result was obtained after $2$ days of
computation or because of a lack of memory.  {\tiny\begin{table}[t!]  \centering {
  \begin{tabular}{|c|c|c||ccccc|}
    \hline
    n & d & k & {\bf RS}  & {\bf RS-T} & {\bf RAG}  & {\bf M} & {\bf T}\\
    3 & 2 & 0 & 1.6  & 8 & 1.6  & 16 & 6.6\\
    4 & 2 & 0 & 1.9  & 10 & 13  & - & - \\
    5 & 2 & 0 & 4.9  & 9 & 329  & - & - \\
    6 & 2 & 0 & 5  & 25 & 1577  & - & - \\
    7 & 2 & 0 & 6  & 1 & 39461  & - & - \\
    8 & 2 & 0 & 10  & 10 & -  & - & - \\
    9 & 2 & 0 & 10  & 13 & -  & - & - \\
    \hline
  \end{tabular}}
  \caption{Results obtained for test-suite {\bf S1}}
  \label{tab:S1}
  \begin{tabular}{|c|c|c||ccccc|}
    \hline
    n & d & k & {\bf RS} & {\bf RS-T} & {\bf RAG}  & {\bf M} & {\bf T}\\
    5 & 3 & 3 & 1762   & - & 1779  & - & - \\
    6 & 3 & 3 & 1583   & - & 376822  & - & -\\
    7 & 3 & 3 & 3135   & - & -  & - & -\\
    8 & 3 & 3 & 4344   & - & -  & - & -\\
    5 & 3 & 4 & 0.4   & - & 0.4  & - & -\\
    6 & 3 & 4 & 0.4   & - & 21  & - &  - \\
    7 & 3 & 4 & 0.6   & - & 440  & - & - \\
    8 & 3 & 4 & 0.9   & - & 11686  & - & - \\
    \hline
  \end{tabular}
  \caption{Results obtained for test-suite {\bf S2}}
  \label{tab:S2}
\end{table}
}Tables~\ref{tab:S1} and \ref{tab:S2} provide the results obtained for the
test-suites {\bf S1} and {\bf S2}. One can see that the use of
Theorem~\ref{thm:main1} allows us to tackle examples that are out of reach of
other implementations.

We also observe that when $2d-1$ becomes larger than $4$ or $5$
implementations based on CAD cannot tackle most of the examples considered
here while the critical point method based implementation {\sc RAGlib}
scales far much better.

\added{Finally, let us consider the following examples extracted from \cite{IMOcomp}:
  \begin{itemize}
  \item the problem {\bf SWE} \cite[p. 98]{IMOcomp} consists in proving that for
    $0<a<b$, the semi-algebraic set defined by
    $m_2>m_1^2\frac{(a+b)^2}{4ab}, (b-x_1)(x_1-a)\geq 0, \ldots,
    (b-x_n)(x_n-a)\geq 0$ is empty with
{\tiny\[
m_1=\frac{1}{n}\sum_{i=1}^n x_i, \quad  m_2=\frac{1}{n}\sum_{i=1}^n x_i^2, \quad a=1, \quad b=2
\]}
\item the problem 3.40.1 (referred to as {\bf ROM}) in \cite[p. 302]{IMOcomp} leads to decide the existence
  of real roots to the semi-algebraic system:
{\tiny\[
\sum_{i<j}x_i (x_i^2+x_j^2)-\frac{1}{8}\left (\sum_{i=1}^nx_i\right )^4>0, x_1>0, \ldots, x_n>0
\]}
  \end{itemize}
  The first (resp. second) table provides timings for {\bf SWE} (resp. {\bf
    ROM}). As observed previously, implementations based on the critical point
  methods scale way better than those based on CAD. Also our approach combined
  with critical point methods allows us to tackle problems which are out of
  reach of the state-of-the-art. {\tiny
  \begin{center}
    \begin{tabular}{c}
    \begin{tabular}{|c|ccccccc|}
      \hline
         n     &3 & 4 &5 &6 & 7 & 8 & 9 \\
    {\bf RS} & 0.12 & 0.14 & 0.3 & 0.5 & 0.6 & 0.74 & 1.2 \\ 
    {\bf RAG} &0.2 & 0.3 & 0.5 &1.6 & 9.8 & 131 & 1978 \\ 
    {\bf M} &0.2 & 14.7 & - & - & - & - & - \\ 
      \hline
    \end{tabular} \\ 
      \\
    \begin{tabular}{|c|ccccccc|}
      \hline
         n     &3 & 4 &5 &6 & 7 & 8 & 9 \\
      \hline
    {\bf RS}   & 0.25 & 0.8 & 4.2 & 95 & 6874 & 6902 & 14023 \\ 
    {\bf RAG} & 0.3 & 1 & 4.5 & 97 & 6664 & - & - \\ 
    {\bf M} & 0.04 & 0.5 & 1246 & - & - & - & - \\ 
      \hline
    \end{tabular}      
    \end{tabular}
  \end{center}}}

\thanks{Cordian Riener is supported by the  Troms\o~ Research Foundation grant 17\_matte\_CR. 
  Mohab Safey El Din is supported by the ANR
  grant ANR-17-CE40-0009 {\sc Galop} and the PGMO grant {\sc Gamma}.}

\bibliographystyle{plain}
\bibliography{symsas}

\end{document}